\documentclass[aps,
 final,%
 notitlepage,%
 twocolumn,
 centertags,
 floatfix, pra]%
 {revtex4-1}
 \usepackage{showkeys}
 \usepackage{graphicx}
 \usepackage{mathtools}
 \usepackage[normalem]{ulem}
 \usepackage{cancel}
 \usepackage{hyperref}
 \usepackage{amsmath,amsthm,amssymb,amscd}
 \usepackage{tikz}
 \usepackage{graphicx,graphbox}
 \usepackage{ulem}
 \usepackage{dsfont}
 
 \usepackage{floatrow}

 \DeclareMathOperator{\Tr}{Tr}
 \usepackage{color}
 \newcommand{\rd}{{\mathrm d}}
 
 \newcommand{\C}{\mathbb C}
 \newcommand{\N}{{\mathbb N}}
 
 \newcommand{\bbone}{{\mathbb I}}

 \newcommand{\e}{{\mathrm e}}
 \newcommand{\Ccal}{{\mathcal C}}

 \newcommand{\rhoG}{\rho_{\textrm G}}

 \newcommand{\LNeg}{{\mathcal{E}_\mathcal{N}}}
 \newcommand{\Ntot}{{\mathcal N}_{\textrm{tot}}}
 \newcommand{\tG}{{\textrm G}}
 
 \newcommand{\vertiii}[1]{{\left\vert\kern-0.25ex\left\vert\kern-0.25ex\left\vert #1 
 		\right\vert\kern-0.25ex\right\vert\kern-0.25ex\right\vert}}

 \newcommand{\tr}{\mathrm{Tr}}
 \newcommand{\Str}{\mathrm{Str}}

 \newcommand{\ket}[1]{\vert #1 \rangle}
 \newcommand{\bra} [1] {\langle #1 \vert}

 \newcommand{\mean}[1]{\langle #1 \rangle}
 
 \newcommand{\psiin}{{\psi_{\textrm{in}}}}
 \newcommand{\psiout}{{\psi_{\textrm{out}}}}

 \newtheoremstyle{break}
 {\topsep}{\topsep}%
 {\itshape}{}%
 {\bfseries}{}%
 {\newline}{}%
 \theoremstyle{break}
 \newtheorem{theorem}{Theorem}

 \newtheorem{corollary}{Corollary}
  \newtheorem{proposition}{Proposition}
 \newtheoremstyle{prime}
 {\topsep}{\topsep}%
 {\itshape}{}%
 {\bfseries}{'}%
 {\newline}{}%
 \theoremstyle{prime}
 \newtheorem{theorembis}{Theorem}

 \newcommand{\Grho}{\rho_{\mathrm{G}}}

 \newcommand{\QCS}{{\mathcal C}}	
 \newcommand{\QCStwo}{{\mathcal C}^2}
 
 \newcommand{\Ccl}{\Ccal_{\textrm{cl}}}
 \newcommand{\EoF}{\mathcal{E}_F}
 
 \newcommand{\Ftot}{{\mathcal F_{\textrm{tot}}}}
 \newcommand{\bR}{{\mathbf{R}}}
 
 \newcommand{\MTN}{{{\mathcal M}_{\textrm{TN}}}}

\begin{document}

 	\title{
 		Relating the Entanglement and Optical Nonclassicality\\ of Multimode States of a Bosonic Quantum Field 
 	}

 	\author{Anaelle Hertz$^{1}$, Nicolas J. Cerf$\,^{2}$, Stephan De Bi\`evre$^{3}$}
 	
 	\address{
 		$^1$ Department of Physics, University of Toronto, Toronto, Ontario M5S 1A7, Canada\\
 		$^2$ Centre for Quantum Information and Communication, \'Ecole polytechnique de Bruxelles,
 		Universit\'e libre de Bruxelles, CP 165, 1050 Brussels, Belgium\\
 		$^3$Univ. Lille, CNRS, UMR 8524, INRIA - Laboratoire Paul Painlev\'e, F-59000 Lille, France
 	}
 	
 %	\date{\today}
 	
 	\begin{abstract}
The quantum nature of the state of a bosonic quantum field may manifest itself in its entanglement, coherence, or optical nonclassicality. Each of these distinct properties is known to be a resource for quantum computing or metrology, and can be evaluated via a variety of measures, witnesses, or monotones. Here, we provide quantitative and computable bounds relating, in particular, some entanglement measures with optical nonclassicality measures. Overall, these bounds capture the fact that strongly entangled states must necessarily be strongly optically nonclassical. As an application, we infer strong bounds on the entanglement that can be produced with an optically nonclassical state impinging on a beam splitter. Then, focusing on Gaussian states, we analyze the link between the logarithmic negativity and a specific nonclassicality measure called “quadrature coherence scale”.

 	\end{abstract}
 	\pacs{vvv}
 	
 	\maketitle
 	\section{Introduction}
 	%\emph{Introduction.} \---- 
 	There are several ways to question the specifically quantum mechanical character of the state of a physical system. 
 	First, one may ask how strongly coherent it is. The existence of coherent superpositions of quantum states is at the origin of interference phenomena in matter waves and, as such, is a typically quantum feature for which several measures and witnesses have been proposed (for a recent review, see \cite{streltsov17}). %(\textcolor{red}{see ? for a recent review}).  
 	Second, when the system under investigation is bi-partite or multi-partite, the entanglement of its components is another intrinsically quantum feature. There exists an extensive literature  exploring a wide variety of measures to quantify the amount of entanglement contained in a given state~\cite{Horodecki1996,Peres1996, Bennett1996, Duan, Simon,  WernerWolf, VidalWerner, Serafini2005, Shchukin, Rudolph, horo42009,  Walborn2009,  Zhang2013}. Finally, for modes of a bosonic quantum field, a third notion of nonclassicality arises, which is often refered to as optical nonclassicality. Following Glauber, the coherent states of an optical field (as well as their mixtures) are viewed as ``classical" as they admit a positive Glauber-Sudarshan P-function~\cite{tigl65}. From there, a variety of measures of optical nonclassicality have been developed over the years, measuring the departure from such optical classical states~\cite{tigl65, hi85, balu86, hi87, hi89,  le91, agta92, le95, luba95, domamawu00, mamasc02, rivo02, kezy04, ascari05, se06, zapabe07, vosp14, ryspagal15, spvo15, kistpl16, al17, na17,  ryspvo17, yabithnaguki18, kwtakovoje19, Bievre19, LuZh19}.

 	Each of these three distinct, typically quantum properties of the state of an optical field have been argued to serve as a resource in quantum information or metrology \cite{yabithnaguki18, kwtakovoje19, FriSkFuDu15, Sahota2015, Ge2018}. The question then naturally arises what the quantitative relations are between these properties. In~\cite{StrSi15}, for example, bounds are given on how much entanglement can be produced from states with a given amount of coherence using incoherent operations: this links coherence with entanglement.
 	In~\cite{HeDeB19}, the coherence and optical nonclassicality of a state are shown to be related to each other: a large value of far off-diagonal density matrix elements $\rho(x, x')$ or $\rho(p,p')$, called ``coherences'', is a witness of the optical nonclassicality of the state. 
 	Our purpose here is to establish a relation between optical nonclassicality and bi-partite entanglement for multi-mode bosonic fields. 
 	
 	One expects on intuitive grounds that a strongly entangled state should be strongly optically nonclassical since all optical classical states are separable.  Conversely, a state that is only weakly optically nonclassical cannot possibly be highly entangled. To make these statements precise and quantitative, we need both a measure of entanglement and one of optical nonclassicality. As a natural measure to evaluate bi-partite entanglement, we use the entanglement of formation ($\EoF$)~\cite{Bennett1996}. Regarding optical nonclassicality, we use a recently introduced monotone~\cite{yabithnaguki18, kwtakovoje19}, which we refer to as the monotone of total noise ($\MTN$). It is obtained by extending to mixed states [through a convex roof construction, see~\eqref{eq:W}] the so-called total noise $\Delta x^2 + \Delta p^2$ defined on pure states, for which it is a well established measure of optical nonclassicality~\cite{yabithnaguki18, kwtakovoje19, Bievre19, LuZh19}. Our first main result (Theorems \ref{thm1} \& \ref{thm:theorembis}')
 	%(see~\eqref{eq:EoFleqW} and~\eqref{eq:EoFleqW3}) provides 
 	consists in an upper bound on $\EoF(\rho)$ as a function of $\MTN(\rho)$ for an arbitrary state $\rho$ of a bi-partite system of $n=n_A+n_B$ modes. In particular, when $n_A=n_B=n/2$, this bound implies that states containing $m$ ebits of entanglement must have an optical nonclassicality -- measured via $\MTN$ -- that grows exponentially with $m$.  As an application, we show that the maximum entanglement that can be produced when a separable pure state impinges on a balanced beam splitter is bounded by the logarithm of the optical nonclassicality of this in-state, measured by $\MTN$. In other words, while it is well known that beam splitters can produce entanglement~\cite{Kim2002, Wang2002, ascari05}, the amount of entanglement  so produced is shown to be severely constrained by the degree of optical nonclassicality of the in-state. 
  More precisely, the amount of nonclassicality needed in the in-state grows exponentially with the number of e-bits of entanglement required at the output. Since it was shown in~\cite{HeDeB19} that strongly optically nonclassical states are extremely sensitive to environmental decoherence which destroys their nonclassicality on a very short timescale,  these results imply it is very difficult to produce strongly entangled states with a beam splitter in the above manner.
 	
 	The  bounds in Theorems \ref{thm1} \& \ref{thm:theorembis}' can readily be computed for pure states since $\EoF$ then coincides with the von Neumann entropy of the reduced state and $\MTN$ coincides with the total noise. For mixed states, however, the bounds relate two quantities that are generally hard to evaluate. Our second main result (Theorem~\ref{thm2}) addresses this issue by considering the special case of (mixed) Gaussian states.  
 	It establishes bounds between explicitly computable measures of entanglement (the logarithmic negativity --  $\mathcal{E}_\mathcal{N}$) and optical nonclassicality (the quadrature coherence scale -- $\QCS$) for Gaussian states. We also derive an explicit simple formula for the quadrature coherence scale of Gaussian states in terms of their covariance matrix [see Eq. \eqref{eq:QCS_Ftot}]. We show that it actually coincides with its Total Quantum Fisher Information ($\mathcal{F}_\text{tot}$), a quantity of importance in metrology which has been shown to provide a nonclassicality monotone~\cite{yabithnaguki18, kwtakovoje19}, albeit not a faithful one.
 	
 	\medskip

 	%%%%%%%%%%%%%%%%%%%%%%%%%
 	%%%%%%%%%%%%%%%%%%%%%%%%%%%
 	
 	\section{Bounding $\EoF$ by optical nonclassicality.}
 	We consider an $n$-mode optical field with annihilation mode operators $a_i=(X_i+ i P_i)/\sqrt{2}$ and corresponding quadratures $X_i$, $P_i$. We set 
 	$\bR=(X_1, P_1, \dots, X_n, P_n)$. The total noise of a pure state $|\psi\rangle$ is defined as
 	$\Ntot(\psi)=\sum_j \Delta R_j^2$ where $\Delta R_j^2$ denotes the variance of $R_j$~\cite{sc86}.
 	%For a general state $\rho$, we consider the convex roof $\MTN$ of $\Ntot$, defined as
 	For a general state $\rho$, the convex roof $\MTN$ of $\Ntot$ is
 	\begin{equation}\label{eq:W}
 	\MTN(\rho)=\frac1n\inf_{\{p_i,\psi_i\}} \sum_i p_i \, \Ntot(\psi_i)\geq 1,
 	\end{equation}
 	where the infimum is over all ensembles $\{p_i,\psi_i\}$ for which $\rho=\sum_i p_i|\psi_i\rangle\langle \psi_i|$, $\sum_i p_i=1$. 
 	It is shown in~\cite{yabithnaguki18, kwtakovoje19} that $\MTN$ belongs to a family of optical nonclassicality monotones and is, as such, a faithful witness of optical nonclassicality: $\MTN(\rho)> 1$ iff $\rho$ is nonclassical. 
 	
 	Now consider a bi-partition of the $n$ modes in two sets of $n_A$ and $n_B$ modes, with $n=n_A+n_B$. We write $\rho_A$ (respectively $\rho_B$) for the reduction of the state $\rho$ to the $n_A$ ($n_B$) modes. If $\rho=|\psi\rangle\langle \psi|$, its entanglement of formation is defined as $\EoF(\psi)=-\Tr\rho_A\ln\rho_A=-\Tr\rho_B\ln\rho_B$. Then, for a general $\rho$, taking the infimum as above~\cite{Bennett1996},
 	$$
 	\mathcal \EoF(\rho)=\inf_{\{p_i,\psi_i\}} \sum_i p_i \, \EoF(\psi_i).
 	$$
 	%where the infimum is taken as above~\cite{Bennett1996}. 
 	We first consider the symmetric case $n_A=n_B=n/2$:
 	\begin{theorem} \label{thm1}Let $\rho$ be a bipartite state with $n_A=n_B=n/2$, then
 		\upshape
 		\begin{equation}\label{eq:EoFleqW}
 		\EoF(\rho)\leq \frac{n}2 \, g\left(\frac{1}2\left(\MTN(\rho)-1\right)\right),
 		\end{equation}
 		\itshape
 		where $g(x)=(x+1)\ln(x+1)-x\ln x$. 
 	\end{theorem}
 	\begin{proof}
 		We first consider pure states $\rho=|\psi\rangle\langle\psi|$. Since both sides of~\eqref{eq:EoFleqW} are invariant under phase space translations, we may assume that  $\langle\psi|R_j|\psi\rangle=0$, $\forall j$. In that case, 
 		$\MTN(\psi)=(2\langle\psi|\hat N|\psi\rangle + n)/n=2N/n+1$, where $N=\langle\psi|\hat N|\psi\rangle$ is the expectation value of the total photon number operator $\hat N=\sum_j a^\dagger_ja_j$ in the centered state $\ket{\psi}$. Similarly, defining $\hat N_A=~\sum_{j=1}^{n_A}a^\dagger_ja_j$ and $\hat N_B=\sum_{j=n_A+1}^{n}a^\dagger_ja_j$, one has $N_A=\Tr\hat N_A\rho_A$, $N_B=\Tr\hat N_B\rho_B$, and $N=N_A+N_B$. Then
 		\begin{eqnarray*}
 			\EoF(\psi)&=&-\Tr \rho_A\ln\rho_A=-\Tr\rho_B\ln\rho_B\\
 			&\leq& \min\left\{n_A \, g\left(\frac{N_A}{n_A}\right) , n_B \, g\left(\frac{N_B}{n_B}\right)\right\},
 		\end{eqnarray*}
 		where $n_A\, g(N_A/n_A)$ is the von Neumann entropy of 
 		%the thermal equilibrium state of $n_A$ modes  with photon number $N_A$, which maximizes the von Neumann entropy at fixed photon number.
 		the product of $n_A$ single-mode thermal states with mean photon number $N_A/n_A$ per mode, which maximizes the von Neumann entropy at fixed mean photon number $N_A$~\cite{We1978}. Maximizing over all states $|\psi\rangle$ with a fixed mean photon number $N$ then yields
 		$$
 		\EoF(\psi)\leq\max_{0\leq N_A\leq N} \min\left\{n_A \, g\left(\frac{N_A}{n_A}\right),n_B \, g\left(\frac{N-N_A}{n_B}\right)\right\}.
 		$$
 		Since $g$ is an increasing function, the maximum is, for each $N$, reached at a unique value $N_A^*$ that depends on $N$ and is the solution of
 		\begin{equation}\label{eq:NAstarbis}
 		n_A \, g\left(\frac{N_A^*}{n_A}\right)=n_B \, g\left(\frac{N-N_A^*}{n_B}\right) .
 		\end{equation}
 		Hence
 		\begin{equation}\label{eq:EoFNAstar}
 		\EoF(\psi)\leq n_A \, g(N_A^*/n_A) := F(N).
 		\end{equation}
 		%  =n_B \, g((N-N_A^*)/n_B)
 		When $n_A=n_B$, then $N_A^*=N_B^*=N/2$, so that 
 		\begin{equation}\label{eq:EoFMTNpure}
 		\EoF(\psi)\leq \frac{n}2 \, g\left(\frac{N}{n}\right) = \frac{n}2  \, g\left(\frac{1}2(\MTN(\psi)-1)\right).
 		\end{equation}
 		This implies Eq.~\eqref{eq:EoFleqW} for any pure state $\rho=|\psi\rangle\langle\psi|$.
 		Now let $\rho$ be an arbitrary state and consider any set of normalized $\ket{\psi_i}$ and $0\leq p_i\leq 1$ such that $\sum_i p_i|\psi_i\rangle\langle\psi_i|=\rho$. 
 		Then Eq.~\eqref{eq:EoFMTNpure} and the concavity of $g$ imply that
 		\begin{eqnarray*}
 			\sum_i p_i \, \EoF({\psi_i})&\leq&\frac{n}2 \sum_i p_i \, g\left(\frac{1}2(\MTN(\psi_i)-1)\right)\\
 			&\leq&\frac{n}2  \, g\left(\frac{1}2 \left(\sum_i p_i \, \MTN(\psi_i)-1\right) \right).
 		\end{eqnarray*}
 		Since $g(x)$ is monotonically increasing, taking the infimum over $\{p_i,\psi_i\}$ on both sides implies Eq.~\eqref{eq:EoFleqW}.
 	\end{proof}
 	One readily sees that, among all states $|\psi\rangle$ with a given $N$, the upper bound $\EoF(\psi)=\frac{n}{2}\, g\left(\frac{N}{n}\right)$ is reached 
 	for an $n/2$-fold tensor product of two-mode squeezed vacuum states with $N/n$ photons per mode,  which is a Gaussian pure state. This is not the unique optimal pure state. We identify all such states in Appendix \ref{s:maximumstate} and show they are not all Gaussian.

 	Since $g(x)$ is an increasing function, the bound~\eqref{eq:EoFleqW}  straightforwardly implies that states with a large entanglement of formations are necessarily strongly nonclassical:  %More precisely we prove the following in Appendix \ref{s:proofcor}:
 	\begin{corollary}
 		\upshape
 		\begin{equation}\label{eq:WEoFexp}
 		\EoF(\rho)\geq \frac{3}{2}\frac{n}{2}\ \ \Rightarrow\ \  \MTN(\rho)\geq 1+ 2 \, \e^{\frac2n\EoF(\rho)-2}.
 		\end{equation}
 		\itshape
 	\end{corollary}
 	\begin{proof}
 		Suppose $\EoF(\rho)\geq \frac32\frac{n}{2}$. Using Eq.~\eqref{eq:EoFleqW}, it implies that  
 		$g\left(\frac{1}2\left(\MTN(\rho)-1\right)\right) \ge \frac32$. Noting that $g(x) \leq x+\frac12$ for all $x$, one can conclude that 
 		%   \leq x+1-\ln(\e-1)
 		% $$
 		% \EoF(\psi)\leq  \frac12\QCStwo(\psi).
 		% $$
 		$\frac{1}{2}(\MTN(\rho)-1) \geq 1$. Now, one also has $g(x)\leq \ln x +1+\frac1x$ for all $x>0$, and hence $g(x)\leq \ln x + 2$ for all $x\geq 1$. 
 		Hence  $g\left(\frac{1}2\left(\MTN(\rho)-1\right)\right) \leq \ln\left(\frac{1}2(\MTN(\rho)-1)\right) + 2 $. Using again Eq.~\eqref{eq:EoFleqW} implies that
 		$
 		\EoF(\rho)\leq\frac{n}2\left( \ln\left(\frac{1}2(\MTN(\rho)-1)\right) + 2\right),
 		$
 		from which one concludes $ \MTN(\rho)\geq1+ 2\e^{\frac2n\EoF(\rho)-2}$,
 		which is Eq.~\eqref{eq:WEoFexp}. 
 	\end{proof}
 	
 	In other words, if we view both entanglement and optical nonclassicality as resources, this inequality shows that the amount of optical nonclassicality of a state $\rho$, as measured by $\MTN(\rho)$, grows exponentially fast with its entanglement of formation, measured in number of ebits.    
 	Conversely, the bound~\eqref{eq:EoFleqW} shows that states with a low optical nonclassicality are necessarily weakly entangled.
 	
 	When $n_A \leq n_B$, Theorem \ref{thm1} can be generalized as follows 
 	%(the proof is given in \cite{supplemental}).
 	\begin{theorembis}\label{thm:theorembis} 
 		Let $\rho$ be a bipartite state with $n_A\leq n_B$, then
 		\upshape
 		\begin{equation}\label{eq:EoFleqW3}
 		\EoF(\rho)\leq n_Ag\left(\frac1{n_A}N_A^*\left(\frac{n}{2}(\MTN(\rho)-1)\right)\right),
 		\end{equation}
 		\itshape
 		where $N_A^*(N)$ is the unique solution of  Eq.~\eqref{eq:NAstarbis}.
 		%\begin{equation}
 		%n_Ag(N_A^*/n_A)=n_Bg((N-N_A^*)/n_B).
 		%\end{equation}
 	\end{theorembis}
 	
 	\begin{proof}
 		We first need to show that the function 
 		$$
 		F(N) := n_A  \, g\left(\frac{N_A^*}{n_A}\right)
 		$$ is concave, where $N_A^*$ is a function of $N$, implicitly defined as the solution of \eqref{eq:NAstarbis},
 		$
 		n_A \, g\left(\frac{N_A^*}{n_A}\right)=n_B \, g\left(\frac{N_B^*}{n_B}\right),
 		$
 		where we defined $N_B^*=N-N_A^*$. Taking the derivative with respect to $N$ in both sides of the last two equalities, one finds 
 		\begin{eqnarray*}
 			g'\left(\frac{N_A^*}{n_A}\right)\frac{\rd N_A^*}{\rd N}&=&g'\left(\frac{N_B^*}{n_B}\right)\frac{\rd N_B^*}{\rd N},\nonumber\\
 			\frac{\rd N_A^*}{\rd N}+\frac{\rd N_B^*}{\rd N}&=&1.
 		\end{eqnarray*}
 		Since $g'>0$, it follows from these two equations that both $\frac{\rd N_A^*}{\rd N}$ and $\frac{\rd N_B^*}{\rd N}$ are positive, so that both $N_A^*$ and $N_B^*$ are increasing functions of $N$. 
 		One readily finds that
 		\begin{eqnarray*}
 			\frac{\rd N_A^*}{\rd N}&=& \frac{g'\left(\frac{N_B^*}{n_B}\right)}{g'\left(\frac{N_A^*}{n_A}\right)+g'\left(\frac{N_B^*}{n_B}\right)},
 			\nonumber\\
 			\frac{\rd N_B^*}{\rd N}&=& \frac{g'\left(\frac{N_A^*}{n_A}\right)}{g'\left(\frac{N_A^*}{n_A}\right)+g'\left(\frac{N_B^*}{n_B}\right)}
 		\end{eqnarray*}
 		and consequently that
 		$$
 		F'(N)= \frac{g'\left(\frac{N_A^*}{n_A}\right)g'\left(\frac{N_B^*}{n_B}\right)}{g'\left(\frac{N_A^*}{n_A}\right)+g'\left(\frac{N_B^*}{n_B}\right)}= \frac{1}{g'\left(\frac{N_A^*}{n_A}\right)^{-1}+g'\left(\frac{N_B^*}{n_B}\right)^{-1}}.
 		$$
 		Now, since $g$ is concave, it follows that $g'$ is a decreasing function of its argument. Since $N_A^*$ is an increasing function of $N$, it then follows that $g'\left(\frac{N_A^*}{n_A}\right)$ is a decreasing function of $N$, and similarly for $g'\left(\frac{N_B^*}{n_B}\right)$. Hence, $F'$ is a decreasing function of $N$, implying that $F$ is concave.

 		We now use this fact to conclude the proof of Theorem~\ref{thm:theorembis}'. We initially follow the same lines as in the proof of Theorem~\ref{thm1}. For  a centered pure state $\psi$, Eq.~\eqref{eq:EoFNAstar} reads
 		$
 		\EoF(\psi)\leq F(N)=F\left(\frac{n}2(\MTN(\psi)-1)\right).
 		$
 		Since both $\EoF$ and $\MTN$  are invariant under phase space translations, one then has, for all $\psi$,
 		$
 		\EoF(\psi)\leq F\left(\frac{n}2(\MTN(\psi)-1)\right).
 		$
 		The concavity of the funtion $F$ further implies that 
 		% as in the proof of Theorem~1,
 		\begin{eqnarray*}
 			\sum_i p_i \, \EoF(\psi_i) &\leq & \sum_i p_i \, F\left(\frac{n}2(\MTN(\psi_i)-1)\right) \nonumber\\
 			&  \leq & F \left( \sum_i p_i \, \frac{n}2\left(\MTN(\psi_i)-1\right) \right).  
 			%&\leq& n_A \sum p_i \, g\left(\frac1{n_A}N_A^*(N_i)\right)\\
 			%&\leq&n_Ag\left(\frac1{n_A}\sum p_i N_A^*(N_i)\right)\\
 			%&= & n_A \, g\left(\frac1{n_A}N_A^*\left(\frac{n}2\sum_i p_i (\MTN(\psi_i)-1)\right)\right)\\
 		\end{eqnarray*}
 		Taking the infimum on both sides and using the fact that $F$ is a monotonically increasing function of its argument, one finds
 		$
 		\EoF(\rho) \leq  F \left(\frac{n}2(\MTN(\rho)-1) \right).   
 		$
 		Recalling the definition of $F$, %in~\eqref{eq:EoFNAstar}
 		 one sees this  is Eq.~\eqref{eq:EoFleqW3}.

 	\end{proof}
 	
 	An analytic expression for $N_A^*(N)$ is not available  when $n_A\not= n_B$, so that \eqref{eq:EoFleqW3} is less explicit than \eqref{eq:EoFleqW}. Nevertheless, for large $N$, one  readily finds the following approximate expression for $N_A^*$ (see Appendix \ref{s:proofthmbis}):
 	\begin{equation}\label{eq:NAstarapprox}
 	N_A^*(N)\simeq(1-\delta)\, N,\ \ \textrm{with}\ \ \delta=\frac{(\e\nu)^{\mu-1}}{\mu(1+(\e\nu)^{\mu-1})},
 	\end{equation}
 	where $\mu=n_A/n_B$ and $\nu=N/n_A$. Consequently, using $g(x)\simeq \ln(x)+1$ for large $x$, one finds
 	approximately that 
 	% \begin{equation}\label{eq:EoFvsMTNasymptold}
 	% \EoF(\rho)\leq n_A\ln\left(\frac{n}{2 n_A}(\MTN(\rho)-1)\right)+n_A,
 	% \end{equation}
 	\begin{equation}\label{eq:EoFvsMTNasympt}
 	\EoF(\rho)\leq n_A\ln\left(\frac{(1-\delta)}{n_A}\frac{n}{2}(\MTN(\rho)-1)\right)+n_A,
 	\end{equation}
 	which is valid for large $\MTN(\rho)$ and shows a similar logarithmic upper bound on $\EoF$ in terms of $\MTN$ as above. This large-$N$ approximation is illustrated in the left panel of Fig.~\ref{fig:BS}.
 	
 	For Gaussian pure states, a simpler and more explicit upper bound can be obtained, which is valid for all values of $N$: 
 	\begin{proposition} 
 			 		 Let $n_A\leq n_B$ and 
 			let $\psi^\tG$ be a pure Gaussian state. Then
 			 		\upshape
 			\begin{eqnarray}\label{eq:EoFvsMTNGauss}
 			\EoF(\psi^\tG)
 			&\leq&n_A  g\left(\frac{n}{4n_A}\left(\MTN(\psi^\tG)-1\right)\right).
 			\end{eqnarray}
 			 		\itshape
 	\end{proposition}
 	\begin{proof}
 		Let us consider a pure Gaussian state $|\psi^{\mathrm G}\rangle$ of an $n=n_A+n_B$ mode system
 		with covariance matrix
 	\begin{equation}\label{eq:matcov}
 		V_{ij}=\mean{\{R_i,R_j\}}-2\mean{R_i}\mean{R_j},
  	\end{equation}
 		where 
 		%\sout{\blue{$\mathbf{R}=~(X_1,P_1,\cdots,X_n,P_n)$ is the quadrature vector and}} 
 		%\textcolor{green}{The vector is already defined on page 2} 
 		$\mean{\cdot}:=\mean{\psi^G|\cdot|\psi^G}$ and $\{\cdot,\cdot\}$ denotes the anticommutator.
 		 We  assume without loss of generality that  the state is centered.  Applying local Gaussian unitaries $U_A^{\tG}, U_B^{\tG}$, such a state can always be transformed into  a state $|\psi_\nu\rangle$, in which Alice and Bob share $n_A$ two-mode squeezed vacuum states, while Bob's remaining $n_B-n_A$ modes are in the vacuum state \cite{Serafini}. Here $\nu=(\nu_1, \nu_2, \dots, \nu_{n_A})$ and the state $|\psi_\nu\rangle$ is characterized by its  covariance matrix $V_\nu$:
 		$$
 		V_\nu=\begin{pmatrix}
 		V_{A,\nu}&C\\C^T&V_{B,\nu}
 		\end{pmatrix},\ \  V_{A,\nu}=\begin{pmatrix}
 		\nu_1\mathds{1}_2&&0\\&\ddots&\\0&&\nu_{n_A}\mathds{1}_2
 		\end{pmatrix}_{2n_A\times2n_A}$$
 		\vspace*{-10pt}
 		$$ V_{B,\nu} = \begin{pmatrix}
 		\nu_1\mathds{1}_2&&0&&\\&\ddots&&0&\\0&&\nu_{n_A}\mathds{1}_2&&\\&&&\mathds{1}_2&&0\\&0&&&\ddots&\\&&&0&&\mathds{1}_2
 		\end{pmatrix}_{2n_B\times2n_B}$$$$ C=\begin{pmatrix}
 		\mu_1\sigma_z&&0&\\&\ddots&&0\\0&&\mu_{n_A}\sigma_z&
 		\end{pmatrix}_{2n_B\times2n_A}
 		$$
 		%$$
 		%V_\nu=\begin{pmatrix}
 		%V_{A,\nu}&C\\C^T&V_{B,\nu}
 		%\end{pmatrix},\  V_{A,\nu}=\begin{pmatrix}
 		%\nu_1\mathds{1}_2&&0\\&\ddots&\\0&&\nu_{n_A}\mathds{1}_2
 		%\end{pmatrix}_{2n_A\times2n_A} \  V_{B,\nu} = \begin{pmatrix}
 		%\nu_1\mathds{1}_2&&0&&\\&\ddots&&0&\\0&&\nu_{n_A}\mathds{1}_2&&\\&&&\mathds{1}_2&&0\\&0&&&\ddots&\\&&&0&&\mathds{1}_2
 		%\end{pmatrix}_{2n_B\times2n_B} C=\begin{pmatrix}
 		%\mu_1\sigma_z&&0&\\&\ddots&&0\\0&&\mu_{n_A}\sigma_z&
 		%\end{pmatrix}_{2n_B\times2n_A}
 		%$$
 		where $\nu_i=\cosh 2r_i$, $\mu_i=\sinh 2r_i$, $\mathds{1}_2=\begin{psmallmatrix}
 		1&0\\0&1
 		\end{psmallmatrix}$ and $\sigma_z=\begin{psmallmatrix}
 		1&0\\0&-1
 		\end{psmallmatrix}.$
 		Since the local unitaries do not change the entanglement of formation, we find
 		$
 		\EoF(\psi^\tG)=\EoF(\psi_\nu)=\sum_{i=1}^{n_A} g\left(\frac12(\nu_i-1)\right)
 		$ \cite{Serafini}.
 		Since $g$ is concave, one has 
 		$$
 		\EoF(\psi^\tG)\leq n_A g\left(\frac1{2n_A}\sum_{i=1}^{n_A}(\nu_i-1)\right).
 		$$
 		On the other hand,
 		\begin{eqnarray*}
 			\Ntot(\psi^\tG)=\frac12\Tr V_\psi
 			&=&\frac12\left(\Tr V_{A,\psi}+\Tr V_{B,\psi}\right)\nonumber\\
 			&\geq& \frac12\left(\Str V_{A,\psi}+\Str V_{B,\psi}\right).
 		\end{eqnarray*}
 	 Here, $\Str V_{A, \psi}$ is the symplectic trace of $V_{A, \psi}$, which is twice the sum of its symplectic eigenvalues; we then have $\tr V_{A,\psi}\geq\Str V_{A,\psi}$ \cite{BhatiaJain}.
 		Since the local symplectic transformations do not change the symplectic spectrum of the reduced states $\rho_A, \rho_B$, we also have
 		\begin{eqnarray*}
 			\Str V_{A,\psi}+\Str V_{B,\psi}&=&\Str V_{A,\nu}+\Str V_{B,\nu}\nonumber\\
 			&=&4\sum_{i=1}^{n_A} \nu_i+2(n_B-n_A).
 		\end{eqnarray*}
 		Hence, 
 		$$\ \ 
 		\Ntot(\psi^\tG)\geq2\sum_{i=1}^{n_A} \nu_i+(n_B-n_A)
 		=2\sum_{i=1}^{n_A} (\nu_i-1)+n_B+n_A.\ \ 
 		$$
 		Since $g$ is monotonically increasing, it follows that 
 		\begin{eqnarray*}
 			\EoF(\psi^\tG)&\leq &n_A  g\left(\frac1{4n_A}\left(\Ntot(\psi^\tG)-n_B-n_A\right)\right)\nonumber\\
 			&=&n_A  g\left(\frac{n}{4n_A}\left(\MTN(\psi^\tG)-1\right)\right).
 		\end{eqnarray*}
 	\end{proof}
 	This is the tightest  possible bound on $\EoF$ for Gaussian pure states that only depends on $\MTN$. 
 	% OLD It is saturated by $n_A$ two-mode squeezed vacuum states (involving the $n_A$ modes of $A$ and the $n_A$ first modes of $B$), with 
 	%the remaining $n_B-n_A$ modes of $B$ in the vacuum. 
  Indeed, one readily checks that it is saturated by $n_A$ two-mode squeezed vacuum states 
 	with identical squeezing parameters 
 	(involving all $n_A$ modes of $A$ and the $n_A$ first modes of $B$), with the remaining $n_B - n_A$ modes of $B$ in the vacuum \emph{i.e.}  if $|\psi^\tG\rangle=|\psi^\nu\rangle$, with $\nu_1=\nu_2=\dots=\nu_{n_A}$.

 	When $n_A=n_B$, the right-hand sides of~\eqref{eq:EoFvsMTNGauss} and~\eqref{eq:EoFMTNpure}  coincide, as expected since the latter inequality is saturated by the above Gaussian pure state. In contrast, as shown in Fig.~\ref{fig:BS}, when $n_A<n_B$ (or $\mu<1$), the right-hand side of~\eqref{eq:EoFvsMTNGauss} is slightly smaller than the one of~\eqref{eq:EoFNAstar}.  It is then  natural to wonder if
 	% by a term $0<\ln(2(1-\delta))$ \blue{\sout{which tends to $\ln2= 0.693$ for large $N$}}\red{Not if $\mu=1$...}. 
 	there are non-Gaussian pure states inside this gap.  This is indeed the case, as we show in Appendix \ref{s:theorembisGaussian}. This means that for a fixed $\MTN$, there exist non-Gaussian pure states with a higher entanglement of formation than any Gaussian pure state with the same value of $\MTN$, provided $n_A\not=n_B$.

 	Note finally that one cannot expect a lower bound on the $\EoF$ in terms of $\MTN$ since a product state has vanishing entanglement while it can have an arbitrarily large $\MTN$. The product of a strongly squeezed pure state with the vacuum is an example of such a case.
 	
 	\begin{figure}[t]
% 		\hspace{-135pt}
% 		\includegraphics[align=t, width=0.53\columnwidth]{gNAversusN_v8_nA3}\hspace{125pt}\includegraphics[trim=0 0cm 0 0cm, clip, align=t, hsmash=c, width=0.4\columnwidth]{EoF_NC_bis_Ln_v2_ancien}	\hspace{-64pt}
% 		\includegraphics[trim=0.5cm 0cm 0 0cm, clip, align=t, width=0.17\columnwidth]{EoF_NC_Bis_Ln_Legend}
% 		\vspace{-8pt}
 			\hspace{-135pt}
 			\includegraphics[align=t, width=0.53\columnwidth]{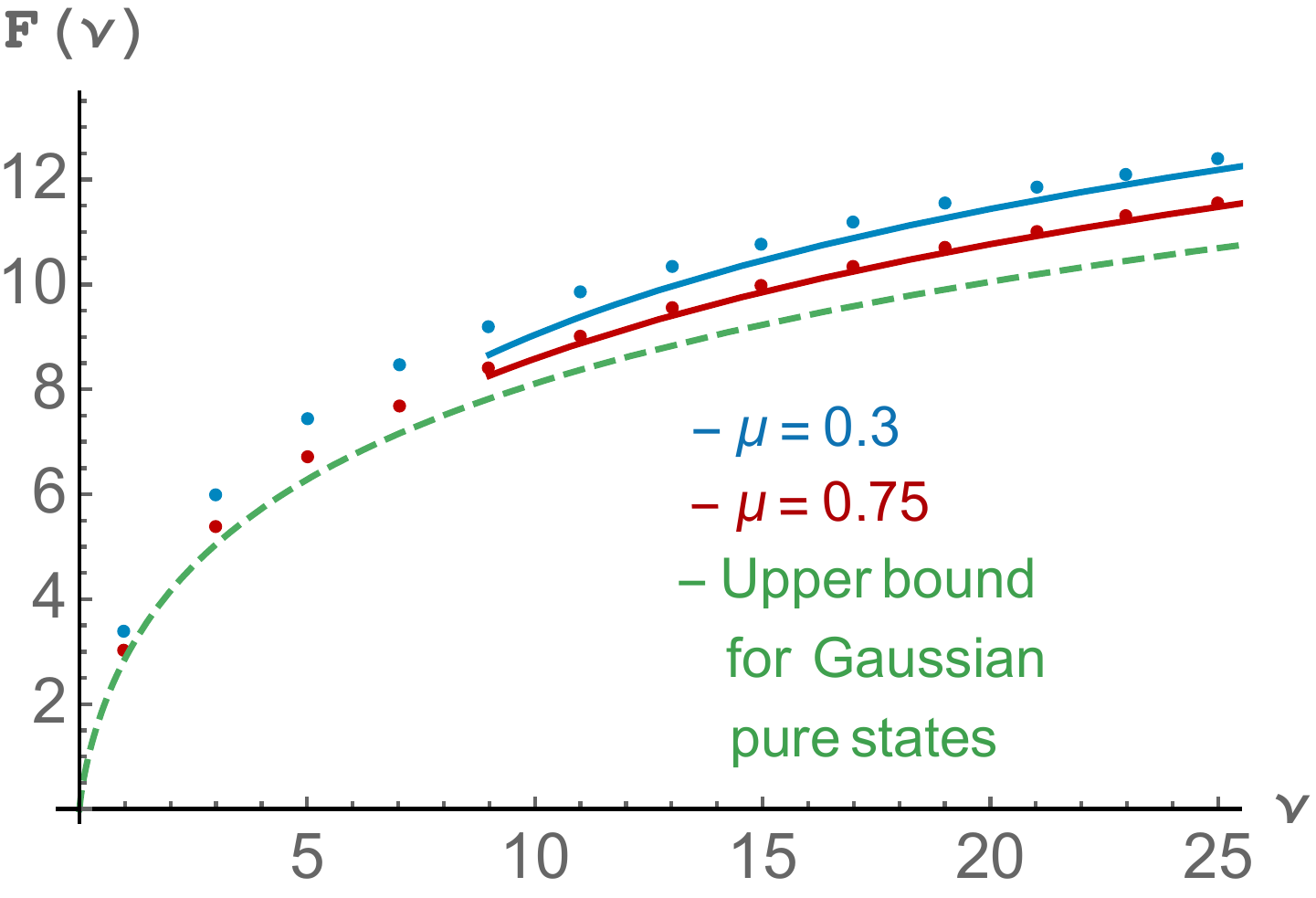}\hspace{125pt}\includegraphics[trim=0 0cm 0 0cm, clip, align=t, hsmash=c, width=0.43\columnwidth]{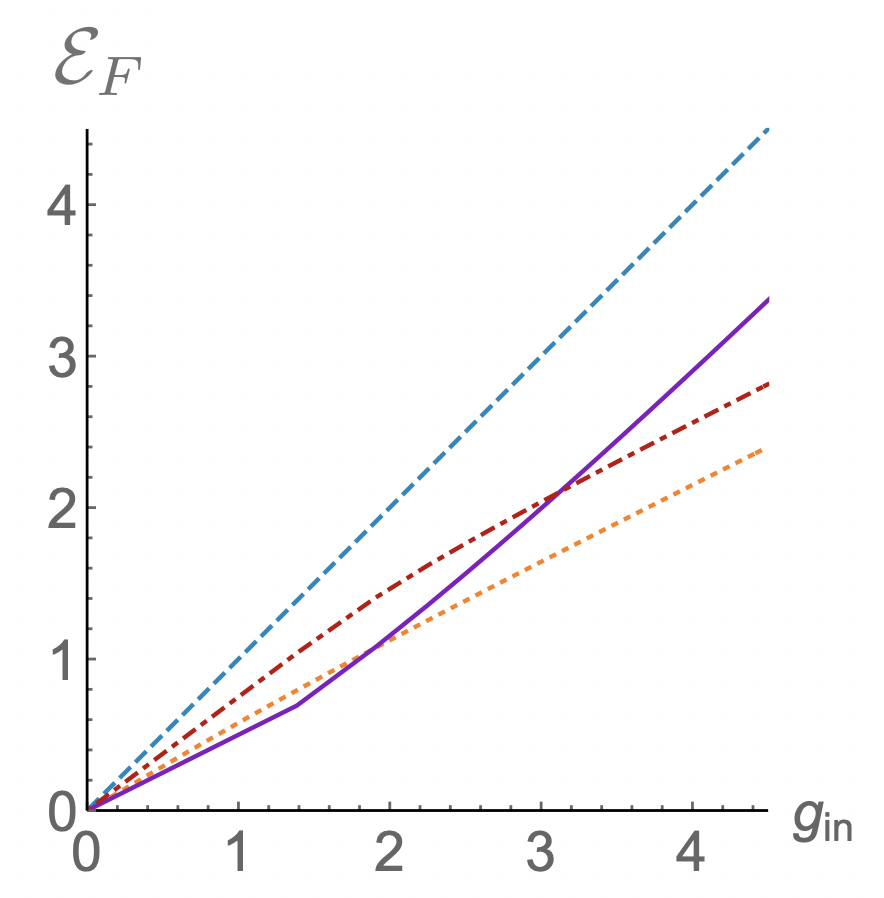}	\hspace{-59pt}
 			\includegraphics[trim=0.5cm 0cm 0 0cm, clip, align=t, width=0.17\columnwidth]{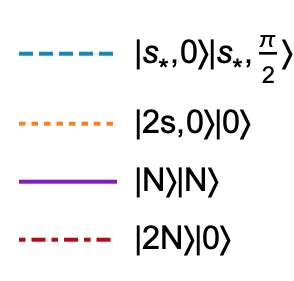}
 			\vspace{-8pt}
 		
 		\caption{ Left: behaviour of  the right hand side of~\eqref{eq:EoFNAstar} as a function of $\nu=N/n_A$, for different values of $\mu=n_A/n_B$, as indicated. The dots are obtained from numerical solutions of~\eqref{eq:NAstarbis}  ($n_A=3$). Full lines are   computed using~\eqref{eq:NAstarapprox} for large $N$.  The green dashed line represents the Gaussian bound in \eqref{eq:EoFvsMTNGauss}.\\
 			Right: $\EoF(\psiout)$ at the output of a beam splitter as a function of $g_{in}=g(\frac12(\MTN(\psiin)-1))$ for various $|\psiin\rangle$, as indicated. 
 			$| 2s,0\rangle|0\rangle$ and $| s_*,0\rangle|s_*,\frac{\pi}{2}\rangle$ are squeezed states with the same %total noise 
 			$\MTN(\psiin)=\cosh(2s_*)$ ($s=\frac14\cosh^{-1}(2\cosh(2s_*)-1)$). The Fock states $\ket{N}\ket{N}$ and $\ket{2N}\ket{0}$ also have the same %total noise
 			$\MTN(\psiin)=2N+1$. 
 			\label{fig:BS}}
 	\end{figure}

 	%%%%%%%%%%%%%%%%%%%%%%%%%%%
 	\section{Entanglement generation with a beam splitter.}
 	It is well known that a balanced beam splitter $\hat B=~\exp(\frac{\pi}{4}(a_1^\dagger a_2-~a_1a_2^\dagger))$  applied to a separable in-state $|\psi_{\textrm{in}}\rangle$ produces an out-state $|\psi_{\textrm{out}}\rangle=\hat B|\psi_{\textrm{in}}\rangle$ that can be entangled provided the in-state is optically nonclassical~\cite{Kim2002, Wang2002, ascari05, kistpl16}.  In~\cite{ascari05} this property is used to quantify the amount of nonclassicality in a single-mode state $|\varphi\rangle$  by the amount of entanglement obtained in the out-state of a balanced beam splitter with as input state $|\psi_{\textrm{in}}\rangle=|\varphi\rangle\otimes|0\rangle$. Here we take a different approach. We treat entanglement and optical nonclassicality as independently and a priori defined properties of the states, and bound the amount of entanglement that can be obtained in the out-state by the amount of nonclassicality of a general separable in-state, as measured by its $\MTN$. By applying Theorem \ref{thm1}, we are indeed able to determine how efficiently a beam splitter can generate entanglement  in this manner. To see this, note that Eq.~\eqref{eq:EoFMTNpure} implies an upper bound on the entanglement of formation of $|\psi_{\textrm{out}}\rangle$ given the amount of optical nonclassicality available in $|\psi_{\textrm{in}}\rangle$, as follows. Let, for any value of the available nonclassicality ${\MTN}_{,0}>0$, 
 	$$
 	S_0=\{|\psiin\rangle=|\varphi_A,\varphi_B\rangle \,\,|\,\,\MTN(\psiin)\leq \MTN_{,0}\}.
 	$$
 	Since $\hat B$ preserves the total noise ($\Ntot(\psiout)=\Ntot(\psiin)$), Eq.~\eqref{eq:EoFMTNpure} implies  
 	$$
 	\EoF(\psiout)\leq g\left(\frac{1}2(\MTN(\psiin)-1)\right)\leq g(\frac12(\MTN_{,0}-1)),
 	$$
 	since $g$ is a monotonically increasing function. 
 	To see the bound is reached, let, for $s\geq0, \phi\in[0,2\pi[$,  $S(s,\phi)=e^{\frac{s}{2}(e^{-i\phi}a^2-e^{i\phi}a^{\dag2})}$ and define $\ket{\varphi_A}=|s_0,0\rangle:=S(s_0,0)|0\rangle$ and $\ket{\varphi_B}=|s_0,\pi/2\rangle:=S(s_0,\frac{\pi}{2})|0\rangle$, with $s_0$ chosen so that $\MTN(\psiin)=\cosh(2s_0)=\MTN_{,0}$. 
 	In this case, $|\psiout\rangle$ $=\hat B|\psiin\rangle=|\psi_{\textrm{TMS}}\rangle$, where $|\psi_{\textrm{TMS}}\rangle$ is the two-mode squee\-zed vacuum state with $\MTN(\psi_{\textrm{TMS}})=\MTN_{,0}$, which we saw saturates~\eqref{eq:EoFMTNpure}. There is a readily identified family of states that saturate the bound (see Appendix~\ref{s:maximumstate}), but typically states in $S_0$ do not. % saturate it.  
 	Several physically interesting examples are given in Fig.~\ref{fig:BS}; see Appendix \ref{s:AppBS} for details on the computations.
 	%details on the computations are provided in~Appendix~\ref{s:AppBS}. 
 	When $|\psiin\rangle=|N,0\rangle$,  $\MTN(\psiin)=N+1$ and the entanglement of formation of the out-state satisfies $\EoF(\psiout)/g(\frac12(\MTN(\psiin)-1))\simeq \frac12$ for large $N$. Hence, only one half of the possible maximal amount of entanglement is produced in this manner for a given amount of optical nonclassicality in the in-state.  
 	When $|\psiin\rangle=|N, N\rangle$, on the other hand, $\MTN(\psiin)=2N+1$, 
 	%the EoF of the out-state is given by the entropy of the distribution $P(m)=\frac1{2^{2N}}\frac{(2N-2m)!(2m)!}{(m!)^2((N-m)!)^2}$, which for large $N$ is approximately given by 
 	and, for large $N$, the entanglement of formation satisfies $\EoF(\psiout)/g(\frac12(\MTN(\psiin)-1))\simeq1$, hence almost the maximum possible amount of entanglement is produced. 
 	It is therefore less efficient to input a $2N$ photon state on one mode and the vacuum on the other, rather than $N$ photons on each.
 	%mode.  
 	%Indeed, in both cases $\MTN(\psiin)=2N+1$, but the output EoF is, for large $N$, twice as large in the second case.  
 	A similar phenomenon occurs with squeezed states at the input: $|\psiin\rangle=|2 s,0\rangle|0\rangle$ and $|\psiin\rangle=|s_*,0\rangle|s_*,\frac{\pi}{2}\rangle$  with $s=\frac14\cosh^{-1}(2\cosh(2s_*)-1)$ have the same values of $\MTN=\cosh(2s_*)$, but the output $\EoF$ is, for large $N$, twice as large in the second case. 
 	%but the corresponding out-state of the first has roughly double the EoF of the second. 
 	
 	Let us point out that, in terms of resource theory, the beam splitter does not ``convert'' nonclassicality into entanglement. Indeed, the total noise 
 	%of the out-state is identical to that of the in-state:
 	is conserved and none of the optical nonclassicality resource is lost in the process. Nevertheless,  the above bounds imply that the in-state must have a  large amount of optical nonclassicality for the entanglement production to be efficient in this manner.  Now note that it was shown in~\cite{HeDeB19} that environmental coupling leads to nonclassicality loss on a time scale inversely proportional to the $\MTN$ of a pure state such as $|\psi_{\textrm{in}}\rangle$: this implies that strongly nonclassical states  $|\psi_{\textrm{in}}\rangle$ with a large $\MTN$ are hard to maintain, so that the production of strongly entangled states with the above procedure will be difficult to realize. 
 	%Statements to the effect that entanglement is easily  synthesized with beam splitters once one disposes of a source of nonclassical light~\cite{yabithnaguki18} should therefore be treated with care.
 	
 	As mentioned above, Theorems \ref{thm1} and \ref{thm:theorembis}' involve convex roofs, which are hard to exploit for mixed states. This is true even for Gaussian states, for which the entanglement of formation remains difficult to evaluate, despite recent progress~\cite{TsRa17,TsOnRa19}. This problem can, however,  be overcome by using alternative, computable measures of entanglement and optical nonclassicality adapted to Gaussian states.

 	%%%%%%%%%%%%%%%%%%%%%%%
 	\section{Gaussian states: bounding $\LNeg$ by $\QCS$.} 
 	Consider a  Gaussian state $\rhoG$ with covariance matrix as defined in \eqref{eq:matcov}.
% 	$$
% 	V_{ij}=\mean{\{R_i,R_j\}}-2\mean{R_i}\mean{R_j},
% 	$$
% 	where 
% 	%\sout{\blue{$\mathbf{R}=~(X_1,P_1,\cdots,X_n,P_n)$ is the quadrature vector and}} 
% 	%\textcolor{green}{The vector is already defined on page 2} 
% 	$\mean{\cdot}:=\tr(\cdot\rhoG)$. 
We will evaluate its nonclassicality using two recently introduced and readily computable quantities: the total quantum Fisher information ($\mathcal{F}_\text{tot}$) (see~\eqref{eq:TQFI}) and the  quadrature coherence scale
 	($\QCS$) (see~\eqref{eq:QCS}). 
 	We recall that the set of optical classical states $\Ccl$~\cite{tigl65} of a system of $n$ modes contains all mixtures of coherent states $D(\alpha)|0\rangle$, where $D(\alpha)=\exp(\alpha a^\dagger-\alpha^*a)$ $[\alpha=(\alpha_1, \alpha_2,\dots,  \alpha_n)\in \C^n \text{ and }a=(a_1,a_2,\cdots,a_n)]$ is the displacement operator and $|0\rangle$ is the $n$-mode vacuum.
 	
  For any state $\rho$ and observable $A$, the quantum Fisher information of $\rho$ for $A$ is
 	$
 	\mathcal F(\rho, A)=4\partial_x^2 D_B^2(\rho, \exp(-ixA)\rho\exp(ixA))_{|x=0},
 	$
 	where $D_B^2(\rho,\sigma)=2(1-F(\rho, \sigma))$ is the Bures distance and $F(\rho,\sigma)=\Tr \sqrt{\sqrt{\rho}\sigma\sqrt{\rho}}$ the fidelity between $\rho$ and $\sigma$. It is known that $\mathcal F(\rho, A)$
 	%, a quantity important in metrology, 
 	is convex in $\rho$ and coincides with $4\, \Delta A^2$ on pure states~\cite{tope13, yu13}. 
The total quantum Fisher information of an $n$-mode state $\rho$ is defined as 
 	\begin{equation}\label{eq:TQFI}
 	\mathcal F_{\textrm{tot}}(\rho)=\frac1{4n}\sum_{j=1}^{2n} \mathcal F(\rho, \bR_j).
 	\end{equation}
 	 It follows that $\mathcal{F}_\text{tot}(\psi)$ coincides with $\MTN(\psi)$ on pure states, and, since it is convex, $\Ftot(\rho)\leq \MTN(\rho)$.
 It is known that the total quantum Fisher information is a nonclassicality  witness, meaning that $\mathcal F_{\textrm{tot}}(\rho)>1$ implies $\rho$ is nonclassical~\cite{Bievre19}.  It is however not a nonclassicality measure since there exist nonclassical states for which $\mathcal F_{\textrm{tot}}(\rho)\leq1$. Contrary to $\MTN$, however, it has the considerable advantage that is can be relatively easily computed on large classes of states. This is in particular true for Gaussian states, where one has~\cite{yabithnaguki18}
 		\begin{equation}\label{eq:FtotTraceV}
 		\Ftot(\rho_G)=\frac1{2n}\tr V^{-1}.
 		\end{equation}

 	In~\cite{Bievre19, HeDeB19} the (squared) quadrature coherence scale  is defined as 
 	\begin{equation}\label{eq:QCS}
 	\QCStwo(\rho)=\frac1{2n\mathcal P}\left(\sum_{j=1}^{2n} \Tr [\rho, R_j][R_j, \rho]\right), \,\,\, \mathcal P=\Tr\rho^2.
 	\end{equation} 
 	The quadrature coherence scale measures the spread of  the coherences of the quadratures of the state~\cite{HeDeB19}. Like the total quantum Fisher information, the quadrature coherence scale is  an optical nonclassicality witness: if $\QCStwo(\rho)>~1$ then $\rho$ is nonclassical~\cite{Bievre19}. But it is not a nonclassicality measure. It can however serve to construct such a measure, as shown in~\cite{Bievre19}.
 	
 	In general, $\mathcal{F}_\text{tot}$ and the $\QCS^2$ capture different properties of states and can in fact strongly differ on certain states~\cite{Bievre19}. Nevertheless, they are both nonclassicality witnesses, and coincide on pure states:
 	\begin{equation}
 	\QCStwo(\psi)=\Ftot(\psi) = \MTN(\psi) = \frac1n \Ntot(\psi) .
 	\end{equation} 
 In addition, as we now show, they coincide on all Gaussian states as well, including mixed ones: 
 	\begin{equation}\label{eq:QCS_Ftot}
 	\QCStwo(\rhoG)=\Ftot(\rhoG)=\frac1{2n}\tr V^{-1}.
 	\end{equation}
  In view of Eq. (\ref{eq:FtotTraceV}), it only remains to prove the first equality.
 	We first note that,  since both $\QCS(\rhoG)$ and $\Ftot(\rhoG)$ are invariant under phase space translations, we can assume that $\mean{R_i}=0$ for all $i$.  The  characteristic function of $\rho_G$ is 
 	$$
 	\chi_G(\xi)=\mathrm{Tr}\rho_G D(\xi)=\exp\{-\frac12\xi^T\Omega V\Omega^T\xi \},$$$$\text{where}\quad\Omega=\bigoplus_{k=1}^{n}\begin{pmatrix}
 	0&1\\-1&0
 	\end{pmatrix},
 	$$ 
 	and $\xi=(\xi_{11},\xi_{12},\cdots,\xi_{n1},\xi_{n2})$. It was shown in ~\cite{gu90, Bievre19} that the  right hand side of~\eqref{eq:QCS} can be written in terms of the characteristic function $\chi(\xi)$ of the state as follows:
 	\begin{equation}\label{eq:QCSWch}
 	\Ccal^2(\rho)=\frac{\||\xi|\chi\|_2^2}{n\|\chi\|_2^2}.
 	\end{equation}
 	Here $\|\cdot\|_2$ designates the $L^2$-norm, (for example $\|\chi\|_2^2:=\int|\chi|^2(\xi)\rd^2 \xi$).
 	From~\eqref{eq:QCSWch}, one finds with a direct computation
 	\begin{eqnarray*}
 		\QCStwo(\rho_G)&=&\frac1{n}\int\left(|\xi_1|^2+\cdots+|\xi_n|^2\right) f(\xi) \rd^2 \xi\\
 		&=&\frac1{n}\tr \Sigma=\frac1{2n}\tr V^{-1}
 	\end{eqnarray*}
 	where $f(\xi)$ is a Gaussian probability density with 0 mean value and covariance matrix $\Sigma=\frac12\Omega V^{-1}\Omega^T$.

 	We will now connect the optical nonclassicality of Gaussian states (measured with $\mathcal{F}_\text{tot}$, or equivalently $\QCS^2$) to their entanglement, measured with the logarithmic negativity $\LNeg$ \cite{VidalWerner}. 
 	%We suppose, without loss of generality, that $n_A\leq n_B$.  
 	Let $\mathcal T_B$ stand for partial transposition in Fock basis applied to the $n_B$ modes only, so that $\tilde \rho =\mathcal T_B[\rho]$ denotes the partial transpose of an arbitrary state $\rho$.
 	% \red{Pourquoi introduire cet opérateur et ne pas juste écrire $\tilde \rho =\rho^{T_B}$? Aussi, tu ne peux pas utiliser $T_B$ ici comme un operateur et le même plus loin comme s'appliquant à la matrice covariance}.
 	It is known that if $\tilde \rho$ is not positive semidefinite, then $\rho$ is entangled~\cite{Peres1996}. For a bi-partite system of $n_A+n_B=n$ modes, the logarithmic negativity of $\rho$ is then defined as
 	$\LNeg(\rho)=\ln(\Tr\sqrt{\tilde\rho^2})$.
 	Note that $\LNeg(\rho)>~0$ implies that $\rho$ is entangled. For pure states, but not in general, one has $\EoF(|\psi\rangle\langle \psi|)\leq\LNeg(|\psi\rangle\langle\psi|)$~\cite{VidalWerner}.
 	The partial transpose $\tilde \rho_G$ of a Gaussian state is again a Gaussian operator, with covariance matrix $\tilde V=T_BVT_B$, where %$T=\text{Diagonal}[\bbone_{2n_1},\bbone_{n_2},-\bbone_{n_2}]$
 	$T_B=\bbone_{2n_A}\bigoplus_{k=1}^{n_B}\begin{psmallmatrix}
 	1&0\\0&-1\end{psmallmatrix}$.  The logarithmic negativity of an arbitrary Gaussian state can be expressed in terms of the symplectic spectrum $\tilde \nu_{-,1}\leq \dots \leq \tilde \nu_{-,n_-} < 1 \leq \tilde \nu_{+,1}\leq \dots \leq\tilde\nu_{+,n_+}$ $(n_{-}+n_+=n$) of $\tilde V$, as follows~\cite{VidalWerner}: if $n_-=0$, then  $\LNeg(\rhoG)=0$, otherwise if $n_-\geq 1$, 
 	\begin{equation*}
 	\LNeg(\rhoG)=\sum_{i=1}^{n_-} \ln\frac{1}{\tilde\nu_{-,i}}.
 	\end{equation*}
 	This equation together with Eq.~\eqref{eq:QCS_Ftot} allow us to derive a main bound for arbitrary Gaussian states:
 	\begin{theorem}\label{thm2}
 		Let $\rho_G$ be a bipartite Gaussian state, then
 		{  \begin{equation}\label{eq:QCStwo_Ftot}
 			\LNeg(\rhoG)\leq n_-\left(\ln \QCStwo(\rho_G)+\ln\frac{n}{n_-}\right).
 			\end{equation} }
 	\end{theorem}
 	
 	\begin{proof}
 		We note that $\tilde V^{-1}=T_BV^{-1}T_B$ so that $\Ccal^2(\rho_G)=\frac1{2n}\tr \tilde V^{-1}$.  Therefore,
 		\begin{eqnarray*}\label{eq:Ccal.vs.ent}
 			\Ccal^2(\rho_G)\geq\frac1{2n}\Str\tilde{V}^{-1}&=&\frac1{ n}\left(\sum_{i=1}^{n_-}\frac1{\tilde\nu_{-,i}}+\sum_{i=1}^{n_+}\frac1{\tilde\nu_{+,i}}\right)\\
 			&\geq&\frac1{ n}\sum_{i=1}^{n_-}\frac1{\tilde\nu_{-,i}}
 		\end{eqnarray*}
 		since $\tr A\geq\Str A$ \cite{BhatiaJain} and since $\tilde\nu_{+,i} \ge 0$, $\forall i$. Using the concavity of the logarithm, it implies~\eqref{eq:QCStwo_Ftot}.
 	\end{proof} 
 	In the special case $n_A=n_B=1$, a better bound can be obtained when $\LNeg(\rhoG)>0$ using the knowledge of $\det V=\det \tilde V=\tilde{\nu}_-^2\tilde \nu_+^2$:
 	\begin{equation*}
 	\Ccal^2(\rho_G)\geq
 	%\frac14\Str\tilde{V}^{-1}&=&
 	\frac12\left(\frac1{\tilde\nu_-}+\frac1{\tilde\nu_+}\right)
 	=\frac12\left(\frac1{\tilde \nu_-}+\frac{\tilde \nu_-}{\sqrt{\det V}}\right)
 	\end{equation*}
 	with  $\tilde\nu_- = \e^{-\LNeg(\rhoG)}$.
 	This inequality is saturated when the trace and symplectic trace of $\tilde V^{-1}$ coincide.

 	Theorem~\ref{thm2} shows that a large entanglement implies a large optical nonclassicality, but, in contrast with Theorem \ref{thm1}, both sides of the inequality are readily computable. 
 	It is instructive to rework~\eqref{eq:QCStwo_Ftot} and eliminate $n_-$ from it:
 	%\begin{corollary}\label{cor:Gaussian}
 	%\end{corollary}
 	\begin{corollary}\label{cor:Gaussian}
 		Let $\rhoG$ be a Gaussian state. Then
 		\begin{eqnarray}\label{eq:QCStwo_Ftot3}
 		\LNeg(\rhoG)>\frac{n}{\mathrm e}&\ \  \Rightarrow\ \ & \ln \QCStwo(\rho_G)\geq \frac1n\LNeg(\rhoG)-\frac{1}{\mathrm e},\qquad
 		\\
 		\label{eq:QCSnminus}
 		\QCStwo(\rho_G)<\e^{-\frac{n}{e}}&\ \Rightarrow\ & \LNeg(\rhoG)=0.
 		\end{eqnarray}
 	\end{corollary}
 	\begin{proof}
 		First note that 
 		$n_-\ln(n/n_-)\leq n/{\mathrm e}$, so that~\eqref{eq:QCStwo_Ftot}
 		implies
 		$
 		n_-\ln \QCStwo(\rho_G)\geq \LNeg(\rhoG)-\frac{n}{\mathrm e}.
 		$
 		Hence, if $ \LNeg(\rhoG)-\frac{n}{\mathrm e}>0$, then $n_-\geq 1$ and  $\QCStwo(\rho_G)>1$. Equation~\eqref{eq:QCStwo_Ftot3}) then follows.
 		%\begin{equation}\label{eq:QCStwo_Ftot3}
 		% \ln \QCStwo(\rho_G)\geq \frac1n\LNeg(\rhoG)-\frac{1}{\mathrm e}>0.
 		% \end{equation}
 		
 		To prove~\eqref{eq:QCSnminus}, note that~\eqref{eq:QCStwo_Ftot} implies that
 		\begin{eqnarray*}\label{eq:QCStwo_Ftot_low}
 			n_-\ln\frac1{\QCStwo(\rho_G)}&\leq& -\LNeg(\rhoG)-n\frac{n_-}{n}\ln\frac{n_-}{n}
 			\leq\frac{n}{e}.
 		\end{eqnarray*}
 		Hence $\QCStwo(\rhoG)\geq \e^{-\frac{n}{\e n_-}}$. Therefore, if $\QCStwo(\rho_G)\leq\e^{-\frac{n}{e}}$ then $n_-<1$ which implies~\eqref{eq:QCSnminus}.
 	\end{proof}
 	Estimate~\eqref{eq:QCStwo_Ftot3} provides a precise quantitative meaning to the statement that a strongly entangled Gaussian state has a large   $\QCStwo(\Grho)$ and is therefore far from optical classicality. One observes here, as in~\eqref{eq:WEoFexp}, an exponential growth of the optical nonclassicality with the entanglement of $\rhoG$.   Estimate~\eqref{eq:QCSnminus} shows that a Gaussian state with small $\QCStwo(\Grho)$ (well below the nonclassicality threshold~$1$) cannot be entangled.

 	%%%%%%%%%%%%%%%%%%%%%%%%%%%%%%%%%%%%

 	\section{Conclusions.}
 	We have established inequalities relating, for arbitrary states of a multi-mode optical field, several standard measures of entanglement and of optical nonclassicality.  In a nutshell,
 	the optical nonclassicality of a strongly entangled state is necessarily large and, in fact, grows exponentially with its entanglement. As an application,  we have bounded the amount of entanglement of formation that can be produced by sending a separable pure state through a beam splitter. Our bound implies that the nonclassicality of the in-state needs to be exponentially large as a function of the expected entanglement of formation of the out-state. Since nonclassicality is a resource that is hard to generate and preserve  due to environmental decoherence, as shown in~\cite{HeDeB19}, our results can be interpreted to say that, inasfar as the states of a multi-mode bosonic quantum field are concerned, the fragility of their entanglement is a consequence of their large nonclassicality.
In addition, we have shown that entanglement is more efficiently produced in a beam splitter when the nonclassicality is distributed equally among the two input modes. 

Measuring entanglement or nonclassicality is, in general, a difficult task, but, by restricting to Gaussian states (including mixed ones),  we also have established bounds between explicitly computable measures. In the process, we have derived an explicit and simple formula for the quadrature coherence scale of a Gaussian state which only depends on the covariance matrix.

Finally,  let us mention that there is interest in comparing the nonclassicality and entanglement of multimode (non-Gaussian) states, for example of the photon added and subtracted states \cite{FabrePRL,FabreNature}. The tools developed in this paper can serve this purpose.

% 	 	{\color{green} \bf We need some sentences on the Gaussian results here, I think.}
%
% 	 	{\color{green} \bf We should mention the papers of Fabre}
% 	 	
% 		{\color{green} \bf We should maybe finish with some open question or further work possibilities?}

 	\acknowledgments
 	%\medskip
 	%\noindent\emph{Acknowledgments.} \----
 	This work was supported in part by the Labex CEMPI (Agence Nationale de Recherche, Grant ANR-11-LABX-0007-01) and by the Nord-Pas-de-Calais Regional Council and the European Regional Development Fund through the Contrat de Projets \'Etat-R\'egion (CPER). The work was also supported by the Fonds de la Recherche Scientifique -- FNRS under Project  No. T.0224.18.  A. H. acknowledge the support of the Natural Sciences and Engineering Research Council of Canada (NSERC).
 	
 	%%%%%%%%%%%%%%%%%%%%%%%%%%%%

 	\appendix

 	%%%%%%%%%%%%%%%%%%%%%%%%%
 	\section{States saturating the bound \eqref{eq:EoFMTNpure}}\label{s:maximumstate}
 	
 	We identify here all pure states $|\psi\rangle$ of $n$ modes that saturate~\eqref{eq:EoFMTNpure} with $n_A=n_B=\frac{n}{2}$; $n$ is  even. 
 	%We set $N=\langle \psi|\hat N|\psi\rangle$. 
 	Let us write
 	\begin{equation}\label{eq:maximizer}
 	|\psi\rangle = \sum_{k,l} c_{k,l} |k,l\rangle,\quad \sum_{k,l}|c_{k,l}|^2=1 ,
 	\end{equation}
 	with $k=(k_1, \dots, k_{n_A}), l=(l_1,\dots, l_{n_A})\in \N^{n_A}$. 
 	The reduced states on the first (or last) $n_A$ modes are
 	$$
 	\rho_A=\sum_{k,k'} \left(CC^\dagger\right)_{k,k'}|k\rangle\langle k'|,\qquad \rho_B=\sum_{l,l'} \left(C^\dagger C\right)_{l,l'}|l\rangle\langle l'|  ,
 	$$
 	where $C$ is the operator 
 	defined as $C=\sum_{k,k'} c_{k,k'} \ket{k}\bra{k'}$ and we use the notation $(\cdot)_{k,k'} = \bra{k} \cdot \ket{k'}$.
 	%acting as follows:
 	%$$
 	%C|k\rangle =\sum_{k'} c_{kk'}|k'\rangle.
 	%$$
 	The right hand side of~\eqref{eq:EoFMTNpure} is the von Neumann entropy of the unique thermal state $\rho_\beta$ of $n_A$ modes, determined by
 	$
 	\rho_\beta=Z_\beta^{-1}\sum_k \e^{-\beta|k|_1}|k\rangle\langle k|,$
 	with $|k|_1=\sum_{i=1}^{n_A} k_i$, where $\beta$ is chosen such that
 	$ \Tr \hat N_A\rho_\beta=\frac{n_A}{e^{\beta}-1}=\frac{N}{2} .
 	$
 	The bound is therefore saturated iff $\rho_A=\rho_B=\rho_\beta$, and hence iff 
 	$C^\dagger C=C C^\dagger=D$, with $D$ being a diagonal operator with entries $d_k=~Z_\beta^{-1}\e^{-\beta|k|_1}$. 
 	Let $U=C D^{-1/2}$, then $U$ is unitary and, with $C=UD^{1/2}$, one finds
 	$
 	D=CC^\dagger=UDU^\dagger \ \Leftrightarrow \  DU=UD.
 	$
 	We conclude that $|\psi\rangle$ in~\eqref{eq:maximizer} saturates the bound iff $C=UD^{1/2}$, with $U$ being a unitary operator commuting with $D$. One obvious choice is to take $U=\bbone$, in which case $|\psi\rangle$ is 
 	%the purification of the $n/2$-mode thermal state with mean photon number $N/n$ per mode, which is a Gaussian state. It is, in fact, 
 	an $n/2$-fold tensor product of two-mode squeezed states 
 	\vspace*{-10pt}
 	$$|\psi_{\textrm{TMS}}\rangle=\cosh(r)^{-1}\sum_{k=0}^{+\infty} (\tanh r)^k \ket{k,k},\quad\beta=\ln\coth^2(r)$$
 	for which $\MTN(\psi_{\textrm{TMS}})=\cosh(2r)$ and $\EoF(\psi_{\textrm{TMS}})=g(\sinh^2(r))$. Thus,  $|\psi\rangle$ is a Gaussian state with $\MTN(\psi)=\cosh(2r)$ and $\EoF(\psi)=\frac{n}{2} \, g(\sinh^2(r))$, saturating the bound~\eqref{eq:EoFMTNpure}.
 	
 	Note that this is not the unique state saturating the bound since such a state is determined by $C=UD^{1/2}=D^{1/2}U$, with $U$ unitary. Therefore, all saturating states can be obtained from the above choice by applying local unitaries $U_A$ and $U_B$ that preserve the photon numbers $\hat N_A$ and $\hat N_B$, setting $C=U_AD^{1/2}U_B$. For example, when $n=2$, they are all states of the form
 	$$
 	|\psi\rangle=\cosh(r)^{-1}\sum_{k=0}^{+\infty} (\tanh r)^k\exp(i\phi_k) \ket{k,k}
 	$$  
 	with arbitrary phases $\phi_k$. If $\phi_k=k\phi$, these states are the general two-mode squeezed states obtained when we inject two orthogonal squeezed states in a balanced beam splitter, the angle of the squeezing of the first input state being $\phi$ (the second input state is squeezed along $\phi+\pi/2$). For general $n$, one has 
 	$$
 	|\psi\rangle=Z_\beta^{-1/2}\sum_{k} \exp(- \beta |k|_1 /2 ) \ket{\varphi_{A,k}, \varphi_{B,k}}
 	$$ 
 	where $|\varphi_{A,k}\rangle = U_A |k\rangle$,  $|\varphi_{B,k}\rangle=U_B|k\rangle$ and $Z_\beta^{-1/2}$ is a normalization constant; note that 
 	$
 	\hat N_A|\varphi_{A,k}\rangle=|k|_1 |\varphi_{A,k}\rangle,$ and $\hat N_B|\varphi_{B,k}\rangle=|k|_1 |\varphi_{B,k}\rangle.
 	$
 	In general, such states are not Gaussian. 
 	
 	%%%%%%%%%%%%%%%%%%%%%%%%%%%%%%%%%%%%%%%5
 	%%%%%%%%%%%%%%%%%%%%%%%%%%%%%%%%%%%%%%
 	%\section{Proof of Corollary 1}\label{s:proofcor}

 	\section{Proof of \eqref{eq:NAstarapprox}}\label{s:proofthmbis}
 	%We may consider $n_A\leq n_B$ with no loss of generality.  
 	%To prove Theorem 1' we first need to show that the function 
 	%$$

 	We prove here the asymptotic expression for $N_A^*(N)$, namely Eq.~\eqref{eq:NAstarapprox}.   We rewrite Eq.~\eqref{eq:NAstarbis} as
 	\begin{equation}\label{eq:NAstartris}
 	\mu g(\nu_*)=g(\mu(\nu-\nu_*)),
 	\end{equation}
 	where $\nu_*=N_A^*/n_A$, $\mu=n_A/n_B$, and $\nu=N/n_A$. 
 	Since we are mostly interested in states with a large optical nonclassicality, we consider the case where $\nu\gg 1$. Writing $\nu_*/\nu=1-\delta$ and using that for large $x$, $g(x)\simeq \ln(\e x)+\frac1{2x}$, Eq. \eqref{eq:NAstartris} becomes 
 	\vspace*{-10pt}
 	$$
 	\mu\ln(\e\nu (1-\delta))+\frac\mu{2\nu(1-\delta)}= \ln(\e\mu\nu\delta)+\frac1{2\mu\nu\delta}.
 	$$
 	Suppose now that $\delta\ll 1$, $\mu\nu\delta\gg1$, and $\nu(1-\delta)\gg1$. Then, we  find
 	\begin{eqnarray}\label{eq:NAstarapprox2}
 	N_A^*&=&(1-\delta)N,\ \\
 	\textrm{with}\ \ \ \delta&=&\frac1{\mu}\left[1-\frac{\e^{1-\mu}}{2\nu^\mu}\right]\frac1{(\e\nu)^{1-\mu}+1}.\nonumber
 	\end{eqnarray}
 	Keeping only the dominant term, one finds~\eqref{eq:NAstarapprox}.
 	In Fig.~\ref{fig:nustargraph} the numerically computed solution to~\eqref{eq:NAstartris} is compared to the asymptotic expressions~\eqref{eq:NAstarapprox2}. The agreement is seen to be excellent, even for relatively small values of $\nu$. The asymptotic expression~\eqref{eq:NAstarapprox} is also shown for comparison.
 	
 	%\begin{figure}
 	%	\floatbox[{\capbeside\thisfloatsetup{capbesideposition={right,bottom},capbesidewidth=6cm}}]{figure}[\FBwidth]
 	%	{\caption{Numerically computed solution $\nu_*$ of~\eqref{eq:NAstartris} (dots) and asymptotic expression~(8) (plain line) and \eqref{eq:NAstarapprox2} (dashed line) of $\nu_*=N_A^*/n_A$, for different values of $\mu$, as indicated.\\}	\label{fig:nustargraph}}
 	%	{\hskip-6cm	\includegraphics[align=t, width=0.76\columnwidth]{nustar_AH_v4}}
 	%\end{figure}
 	
 	\begin{figure}
 		\centerline{
 			\hskip-6cm \includegraphics[height=4cm, keepaspectratio]{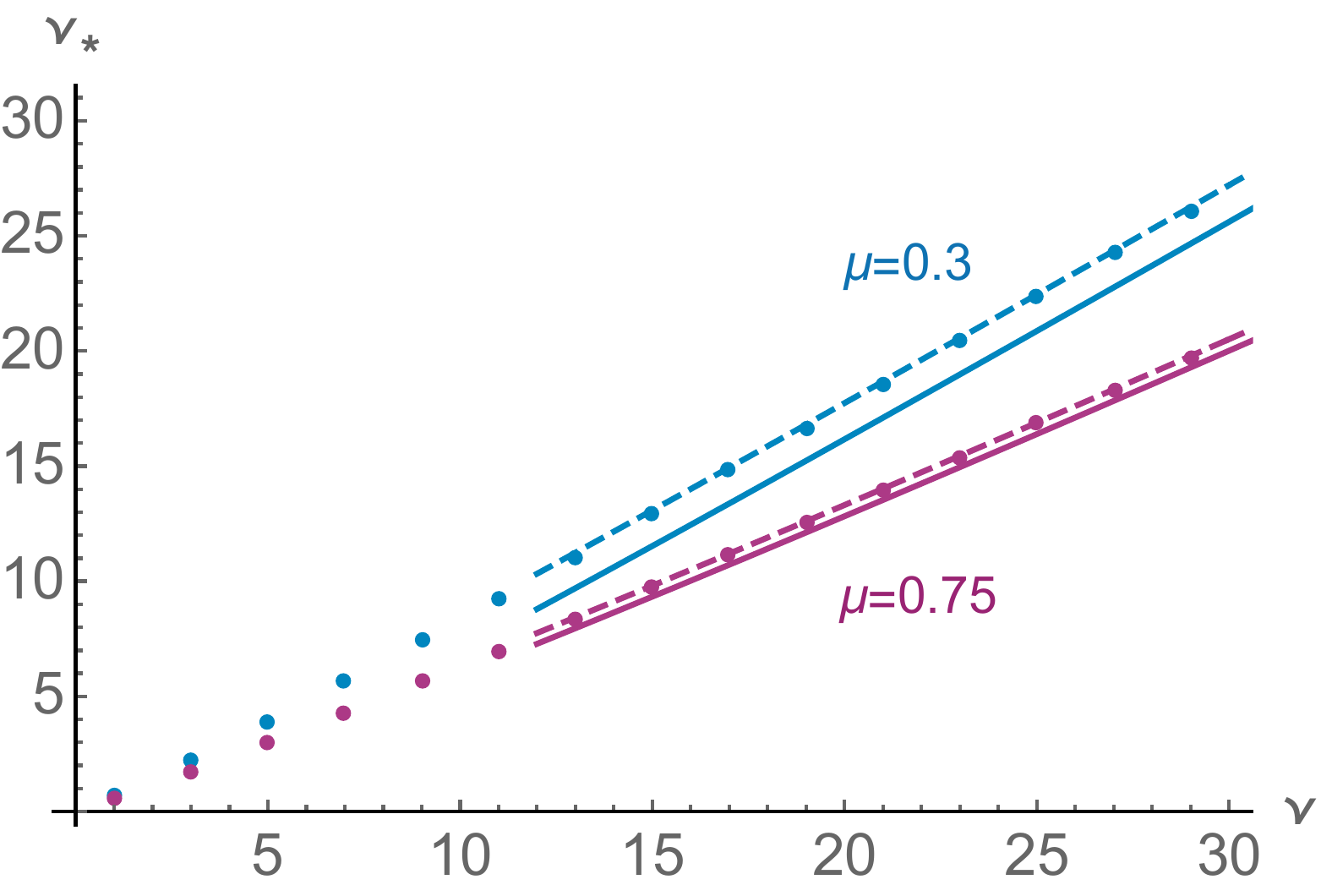}
 		}
 		\caption{Numerically computed solution $\nu_*$ of~\eqref{eq:NAstartris} (dots) and asymptotic expression~(8) (plain line) and \eqref{eq:NAstarapprox2} (dashed line) of $\nu_*=N_A^*/n_A$, for different values of $\mu$, as indicated. %{\color{blue} Anaelle, please do. I decided finally it is worth doing, since it corroborates the quality of the approximation given by the formula. See non-professional figure above. Note that delta1 correspond to only the dominant term in delta (the one written in (8), whereas delta2 takes the extra term above into account. I suggest only to plot delta2. Perhaps for several values of mu.}
 			\label{fig:nustargraph}}
 	\end{figure}

 	%%%%%%%%%%%%%%%%%%%%%%%%%%%%%%
 	\section{Non-Gaussian states violating bound~\eqref{eq:EoFvsMTNGauss}}\label{s:theorembisGaussian}

 	We now show that there exist non Gaussian states that violate the Gaussian bound~\eqref{eq:EoFvsMTNGauss}. For that purpose, we consider the case $n_A=1, n_B=2$. The Gaussian states that saturate the bound are then explicitly given by $(0<~q<1)$
 	\begin{equation*}
 	|\psi_q\rangle=(1-q)^{1/2}\sum_n q^{n/2} |n;n,0\rangle.
 	\end{equation*}
 	
 	Here we wrote $|n;m_1,m_2\rangle$ for the Fock state with $n$ photons in the single $A$ mode and $m_1$, respectively $m_2$ photons in the two modes of $B$. For this state, we have explicitly
 	$
 	\langle \psi_q|\hat N|\psi_q\rangle=\frac{2q}{1-q}$ and $ \MTN(\psi_q)=\frac{2}{3}\langle \psi_q|\hat N|\psi_q\rangle+1.
 	$
 	We will now exhibit a non Gaussian local transformation $U_B$ that, when applied to $|\psi_q\rangle$, yields a state $|\psi'\rangle=U_B|\psi_q\rangle$ that has the same entanglement of formation as $|\psi_q\rangle$ (since $U_B$ is local) but that lowers its $\MTN$. In other words, we will show that
 	\begin{equation}\label{eq:mtnbreak}
 	\MTN(\psi')<\MTN(\psi_q).
 	\end{equation}
 	This implies that
 	$$\ 
 	\EoF(\psi')=\EoF(\psi_q)=g(\langle \psi_q|\hat N |\psi_q\rangle/2)>g(\langle \psi'|\hat N |\psi'\rangle/2),
 	\ $$
 	and therefore shows $|\psi'\rangle$ does not satisfy the Gaussian bound~\eqref{eq:EoFvsMTNGauss}. Of course, it does satisfy the bound~\eqref{eq:EoFleqW}. The local transformation $U_B$ is constructed as follows. Let $k> 1$ be fixed. Then 
 	\begin{eqnarray*}
 		U_B|m_1, m_2\rangle&=&|m_1, m_2\rangle, \ \ \textrm{if}\ m_1m_2\not=0,\nonumber\\
 		U_B|0,0\rangle&=&|0,0\rangle,\\
 		U_B|m_1, 0\rangle&=&|m_1,0\rangle, \ \ \textrm{if}\ m_1\not=k,\\
 		U_B|k, 0\rangle&=&|0,1\rangle,\\
 		U_B|0,1\rangle&=&|k,0\rangle\\
 		%U_B|0, m_2\rangle&=&|0, m_2+1\rangle, \ \textrm{if}\ m_2\red{>1}.
 		U_B|0, m_2\rangle&=&|0, m_2\rangle, \ \ \textrm{if}\ m_2>1.
 	\end{eqnarray*}
 	With $|\psi'\rangle=U_B|\psi_q\rangle$ one then easily checks that, for $i=1,2,3$, 
 	$
 	\langle \psi'|a_i|\psi'\rangle=0=\langle \psi'|a_i^\dagger|\psi'\rangle,
 	$
 	so that $\langle \psi'|X_i|\psi'\rangle=0=\langle \psi'|P_i|\psi'\rangle$. It follows that
 	$
 	\MTN(\psi')=\frac{2}{3}\langle \psi'|\hat N|\psi'\rangle+1.
 	$
 	It will therefore suffice to prove $\langle \psi'|\hat N|\psi'\rangle<\langle \psi_q|\hat N|\psi_q\rangle$. One readily finds
 	$
 	\langle \psi'|\hat N|\psi'\rangle=\langle \psi|\hat N|\psi\rangle+(1-q)q^k(1-k).
 	$
 	so that~\eqref{eq:mtnbreak} follows since $k>1$. 
 	
 	%%%%%%%%%%%%%%%%%%%%%%%%%%%%%%%%
 	%%%%%%%%%%%%%%%%%%%%%%
 	%%%%%%%%%%%%%%%%%%%%%%%%%%%%
 	\section{Beam splitters}\label{s:AppBS}
 	For many states commonly considered the entanglement produced by the beam splitter is considerably lower than the maximal value possible.
 	For example, when $|\psiin\rangle=|N,0\rangle$, one obtains 
 	\begin{equation*}\label{transFock}
 	|\psi_{\textrm{out}}\rangle=\hat B|N,0\rangle=\sum_{m=0}^N\sqrt\frac{n!\,2^{-N}}{m!(N-m)!}|m, N-m\rangle.
 	\end{equation*}
 	Then  $\MTN(\psiin)=N+1$ and $\EoF(\psiout)$ is given by the entropy of the binomial distribution $P(k)=\frac{N!}{k!(N-k)!}2^{-N}$, which for large $N$ is approximately given by $\EoF(\psiout)\simeq\frac12\ln(2\pi\e N)$. Hence, in this case, $\EoF(\psiout)/g(\frac12(\MTN(\psiin)-1))\simeq \frac12$, as can be observed in Fig.~\ref{fig:BS}.
 	When $|\psiin\rangle=|{N}, {N}\rangle$, one  has $\MTN(\psiin)=2N+1$ and \cite{Kim2002}
 	$$
 	|\psiout\rangle=\hat B\ket{N,N}=\sum_{m=0}^Nc_m|2m, 2N-2m\rangle,
 	$$$$\text{with}\ \ \
 	c_m=\frac1{2^{N}}\frac{\sqrt{(2N-2m)!(2m)!}}{m!(N-m)!}.
 	$$
 	For large $N$, choosing $m=N/2+\delta$, one can apply the Stirling approximation $N!\rightarrow\sqrt{2\pi N}(N/e)^N$ to find that the coefficients $|c_m|^2$ converge to $\frac{f(m/N)}{N}$ with 
 	$f(x)=~\frac1{\pi}\frac1{\sqrt{\frac{m}{N}(1-\frac{m}{N})}}$. This result coincides with the one obtained in \cite{Nakazato2016}. The Von Neumann entropy is thus given by 
 	\begin{eqnarray*}
 		-\tr \rho_1\ln\rho_1&=&-\sum_{m=0}^N|c_m|^2\ln |c_m|^2\\
 		&=&-\frac1{N}\sum_{m=0}^NN|c_m|^2\ln |c_m|^2\\
 		&=&-\frac1{N}\sum_{m=0}^Nf(m/N)\ln f(m/N)+\ln N\nonumber\\
 		&\approx&-\int_0^1f(x)\ln f(x)\text{d}x+\ln N\\
 		&=&\ln\frac{\pi}{4}+\ln N.
 	\end{eqnarray*}
 	Hence, in this case $\EoF(\psiout)/g(\frac12(\MTN(\psiin)-1))\simeq 1$ which means that,  
 	asymptotically, the maximal possible amount of entanglement can be produced in this manner.
 	It is therefore more efficient to input a state with $N$ photon on each mode than $2N$ photons on one mode and the vacuum on the other, as in both cases $\MTN(\psiin)=2N+1$, but the output $\EoF$ is, for large $N$, twice as large in the first case.
 	
 	If $|\psiin\rangle=|2 s,0\rangle|0\rangle$, one finds the out-state is a two-mode squeezed vacuum state  of parameter $s$ on which we add some squeezing $s'$ on the first mode and $-s'$ on the second. Since those squeezing are local, they do not modify the value of the entanglement of formation, which is thus the one of the two-mode squeezed vacuum state . Hence, $\EoF(\psiout)=g(\sinh^2(s))$. Note, nevertheless, that while a two-mode squeezed vacuum state  of parameter $s$ has a total noise of $\cosh(2s)$, the total noise of the in-state is given by $\MTN(\psiin)=\frac{\cosh(4s)+1}{2}$. Only about one half of the possible maximal amount of entanglement is produced in this manner.
 	%A similar phenomenon occurs with squeezed states. If $|\psiin\rangle=|2 s,0\rangle|0\rangle$, one finds $\MTN(\psiin)=\frac{\cosh(4s)+1}{2}$, while $\EoF(\psiout)=g(\sinh^2(s))$.
 	On the other hand,  if $|\psiin\rangle=|s_*,0\rangle|s_*,\frac{\pi}{2}\rangle$, with  $s=\frac14\cosh^{-1}(2\cosh(2s_*)-1)\simeq \frac{s_*}{2}$ the total noise of the in-state is also given by $\MTN(\psiin)=\frac{\cosh(4s)+1}{2}$ but yields, after the beam splitter,  the maximum entanglement possible, namely $g(\sinh^2(s_*))\geq g(\sinh^2(s))$. So in this instance too it is more efficient, in terms of entanglement creation, to insert a symmetric input in the beam splitter .

 	%%%%%%%%%%%%%%%%%%%%%%%%%%%%%%%

% 	\bibliographystyle{unsrt}
% 	\bibliographystyle{apsrev4-1}
% 	\bibliography{NCvsEoF_references}
 	
 	%merlin.mbs apsrev4-1.bst 2010-07-25 4.21a (PWD, AO, DPC) hacked
 	%Control: key (0)
 	%Control: author (72) initials jnrlst
 	%Control: editor formatted (1) identically to author
 	%Control: production of article title (-1) disabled
 	%Control: page (0) single
 	%Control: year (1) truncated
 	%Control: production of eprint (0) enabled
 	%

 	%%%%%%%%%%%%%%%%%%%%%%%%%%%%%%%
 	%%%%%%%%%%%%%%%%%%%%%%%%%%%%%%%%%%

 \end{document}